\documentclass[a4paper,11pt]{article}

\usepackage[colorlinks=true]{hyperref}

\usepackage[top=3cm, bottom=4cm, left=4cm, right=4cm]{geometry}

\usepackage{amsmath}
\usepackage{amssymb}
\usepackage{amsthm}

\usepackage[footnotesize]{caption}
\usepackage{tikz}
\usetikzlibrary{shapes}

\usepackage{graphicx}
\usepackage{subcaption}

\usepackage{authblk}

\usepackage{breakurl}

\newtheorem{theorem}{Theorem}[section]

\newtheorem{lemma}[theorem]{Lemma}
\newtheorem{corollary}[theorem]{Corollary}

\theoremstyle{definition}
\newtheorem{definition}[theorem]{Definition}

\newtheorem{example}[theorem]{Example}
\newtheorem*{definition*}{Definition}
\newtheorem*{construction*}{Construction}

\DeclareMathOperator{\sdgon}{sdgon}

\tikzstyle{vertex}=[circle, draw, fill=black, inner sep=0pt, minimum width=4pt]
\tikzstyle{added}=[diamond, fill=gray, inner sep=1.5pt, minimum width=4pt]
\tikzstyle{edge} = [line width = 1pt]

\title{Stable divisorial gonality is in NP}

\author[1, 2]{\small Hans L.\ Bodlaender\thanks{This author was partially supported by the Networks project, founded by the Dutch Organization for
Scientific Research NWO.}}
\author[1]{\small Marieke van der Wegen}
\author[1]{\small Tom C.\ van der Zanden}

\affil[1]{\footnotesize Department of Information and Computing Sciences, Universiteit Utrecht, Princetonplein 5, 3584 CC Utrecht, The Netherlands}
\affil[2]{\footnotesize Department of Mathematics and Computer Science, Eindhoven University of Technology, PO Box 513, 5600 MB Eindhoven, The Netherlands}

\date{}

\begin{document}
	\maketitle
	
	\begin{abstract}
Divisorial gonality and stable divisorial gonality are graph parameters, which have an origin in algebraic geometry. Divisorial gonality
of a connected graph $G$ can be defined
with help of a chip firing game on $G$. The stable divisorial gonality of $G$ is the minimum divisorial gonality over all subdivisions of edges
of $G$.

		In this paper we prove that deciding whether a given connected graph has stable divisorial gonality at most a given integer $k$ belongs to the class NP. Combined with the result that (stable) divisorial gonality is NP-hard by Gijswijt, we obtain that stable divisorial gonality is NP-complete. The proof consist of a partial certificate that can be verified by solving an {Integer Linear Programming} instance. As a corollary, we have that
the number of subdivisions needed for minimum stable divisorial gonality of a graph with $n$ vertices is bounded by $2^{p(n)}$ for a polynomial $p$.
	\end{abstract}

	\section{Introduction}
	The notions of the divisorial gonality and stable divisorial gonality of a graph find their origin in algebraic geometry and are related to the abelian sandpile model (cf.\ \cite{boek}). The notion of divisorial gonality 
was introduced by Baker and Norine~\cite{Baker,BakerN07}, under the name gonality. As there are several
 different notions  of gonality in use (cf.\ \cite{Baker,Caporaso,CKK}), we add the term \textit{divisorial}, following \cite{Caporaso}. See \cite[Appendix A]{CKK} for an overview of the different notions. 
 
Divisorial gonality and stable divisorial gonality have definitions in 
terms of a chip firing game. In this chip firing game, played on a connected multigraph $G =(V,E)$, each vertex has a non-negative number of chips. 
When we {\em fire} a set of vertices $S\subseteq V$, we move from each vertex $v\in S$ one chip over each edge with $v$ as endpoint. Each vertex $v$
in $S$ has its number of chips decreased by the number of edges from $v$ to a neighbour not in $S$, and each vertex $v$ not in $S$ has its number of chips increased
by the number of edges from $v$ to a neighbour in $S$.
Such a firing move is only allowed when after the move, each vertex still has a nonnegative number of chips. The {\em divisorial gonality} of
a connected graph $G$ can be defined as the minimum number of chips in an initial assignment of chips (called {\em divisor}) such that for each vertex $v\in V$, there
is a sequence of allowed firing moves resulting in at least one chip on $v$. Interestingly, this number equals the number for a {\em monotone}
variant, where we require that each set that is fired has the previously fired set as a subset. See Section~\ref{section:definitions} for
precise definitions.

A variant of divisorial gonality is {\em stable divisorial gonality}. The stable divisorial gonality of a graph is the minimum of
the divisorial gonality over all subdivisions of a graph; we can subdivide the edges of the graph any nonnegative number of
times. (In the application in algebraic geometry, the notion of {\em refinement} is used. Here, we can subdivide
edges but also add new degree one vertices to the graph in a refinement, but as this
never decreases the number of chips needed, we can ignore the possibility of adding leaves. Thus, we use subdivisions 
instead of refinements.) 

In this paper, we study the complexity of computing the stable divisorial gonality of graphs: i.e., we look at the complexity of the \textsc{Stable Divisorial Gonality} problem:
given an undirected graph $G=(V,E)$ and an integer $K$, decide whether the stable divisorial gonality of $G$ is at most $K$.
It was shown by  Gijswijt~\cite{Gijswijt} that divisorial gonality is NP-complete. The same reduction gives that stable divisorial gonality
is NP-hard. However, membership of stable divisorial gonality in NP is not trivial: it is unknown how many subdivisions are needed
to obtain a subdivision with minimum divisorial gonality. In particular, it is open whether a polynomial number of edge subdivisions are
sufficient. 

In this paper, we show that stable divisorial gonality belongs to the class NP. We use the following proof technique, which we think is interesting in its own right: we give partial certificates
that describe only some aspects of a firing sequence. Checking if a partial certificate indeed corresponds
to a solution is non-trivial, but can
done by solving an integer linear program. Membership in NP follows by adding to the partial certificate, that describes aspects of the firing sequence,
a certificate for the derived ILP instance.
As a corollary, we have that
the number of subdivisions needed for minimum stable divisorial gonality of a graph with $n$ vertices is bounded by $2^{p(n)}$ for a polynomial $p$.

It is known that treewidth is a lower bound for stable divisorial gonality \cite{JosseGijswijtTreewidth}. The stable divisorial gonality of a graph is at most the divisorial gonality, but this inequality can be strict, see for example \cite[Figure 1]{BodewesBCW17}. 

We finish this introduction by giving an overview of the few previously known results on the algorithmic complexity of (stable) divisorial gonality.
Bodewes et al.~\cite{BodewesBCW17} showed
that deciding whether a graph has stable divisorial gonality at most 2,
and whether it has divisorial gonality at most 2 can be done in $O(n \log n +m)$ 
time. From \cite{Bruyn2012reduced} and \cite{BakerShokrieh}, it follows that divisorial gonality
belongs to the class XP. It is open whether stable divisorial gonality is in XP.
NP-hardness of the notions was shown by Gijswijt \cite{Gijswijt}.

	\section{Preliminaries}
	\label{section:definitions}
In this paper, we assume that each graph is a connected undirected multigraph, i.e., we allow parallel
edges. In the algebraic number theoretic application of (stable) divisorial gonality, graphs can also
have selfloops (edges with both endpoints at the same vertex), but as the (stable) divisorial gonality of graph
does not change when we remove selfloops, we assume that there are no selfloops.

	A \emph{divisor} $D$ is a function $D\colon V(G) \to \mathbb{Z}$. The \emph{degree} of a divisor is $\deg(G) = \sum_{v\in V} D(v)$. 
    We call a divisor \emph{effective} if $D(v) \geq 0$ for all vertices $v$. 
     Let $D$ be an effective divisor and $A$ a set of vertices. We call $A$ \emph{valid}, if for all vertices $v \in A$ it holds that $D(v)$ is at least the number of edges from $v$ to a vertex outside $A$. When we \emph{fire} a set $A$, we obtain a new divisor: for every vertex $v\in A$, the value of $D(v)$ is decreased by the number of edges from $v$ to vertices outside $A$ and for every vertex $v\notin A$, the value $D(v)$ is increased by the number of edges from $v$ to $A$. We are only allowed to fire valid sets, so that the divisor obtained is again effective.

    Two divisors $D$ and $D'$ are called {\em equivalent}, if there is an increasing sequence of sets $A_1 \subseteq A_2 \subseteq \ldots \subseteq A_k \subseteq V$ such that for every $i$ the set $A_i$ is valid after we fired $A_1, A_2, \ldots, A_{i-1}$ starting from $D$, and firing $A_1, A_2, \ldots, A_k$ yields $D'$. We write $D \sim D'$ to denote that two divisors are equivalent. For two equivalent divisors $D$ and $D'$, the difference $D' - D$ is called \emph{transformation} and the sequence $A_1, A_2, \ldots, A_k$ is called a \emph{level set decomposition} of this transformation. 
A divisor $D$ \emph{reaches} a vertex $v$ if it is equivalent to a divisor $D'$ with $D'(v) \geq 1$. 
  
A subdivision of a graph $G$ is a graph $H$ obtained from $G$ by applying a nonnegative number of times the following operation: take an edge between two vertices $v$ and $w$ and replace this edge by two edges to
a new vertex $x$.

The \emph{stable divisorial gonality} of a graph $G$ is the minimum number $k$ such that there exists a subdivision $H$ of $G$ and a divisor on $H$ with degree $k$ that reaches all vertices.  

There are several equivalent definitions, which we omit here. If we do not require that the sequence of
firing sets is increasing, i.e., we omit the requirements $A_i \subseteq A_{i+1}$, then we still have the
same graph parameter (see \cite{Bruyn2012reduced}). The notion of a firing set can be replaced by an algebraic operation
(see \cite{BakerN07}); instead of subdivisions, we can use {\em refinements} where we allow that we add subdivisions and 
trees, i.e., we can repeatedly add new vertices of degree one. The definition we use here is most intuitive and useful for our proofs.

	\section{A (partial) certificate}

	For a yes-instance $(G,k)$ of the problem, there is a subdivision $G'$ and a divisor $D$ with $k$ chips that reaches all vertices. We do not know whether the number of subdivisions in $G'$ is polynomial in the size of the graph, i.e.\ in the number of vertices and edges of the graph, so we cannot include $G'$ in a polynomial certificate for this instance. But the chips in $D$ can be placed on added vertices of $G'$, so we cannot include $D$ in our certificate either. We will prove that when we subdivide
every edge once, we can assume that there is a divisor $D'$ that reaches all vertices and has all chips on vertices of this new graph, and hence we can include $D'$ in a polynomial certificate.

	\begin{definition}
		Let $G$ be a graph. Let $G_1$ denote the graph obtained by subdividing every edge of $G$ once. 
	\end{definition}
	
	\begin{lemma}\label{lem:chipsOnRefinement}
		Let $G$ be a graph. The stable divisorial gonality of $G$ is at most $k$ if and only if there is a subdivision
        $H$ of $G_1$ and a divisor $D$ on $H$ such that \begin{itemize}
			\item $D$ has degree at most $k$, 
			\item $D$ reaches all vertices of $H$,
			\item $D$ has only chips on vertices of $G_1$. 
		\end{itemize}
	\end{lemma}
	\begin{proof}
		Suppose that there exists a subdivision $H$ of $G_1$ and a divisor with the desired properties. Then it is clear that the stable divisorial gonality of $G$ is at most $k$, since $H$ is a subdivision
of $G$ as well. 
		
		Suppose that $G$ has stable divisorial gonality at most $k$. Then there is a subdivision
$H$ of $G$ and a divisor $D$ on $H$ with degree at most $k$ that reaches all vertices. If not every edge of $G$ is subdivided in $H$, then subdivide every edge of $H$ to obtain $H_1$. Consider the divisor $D$ on $H_1$. By \cite[Corollary 3.4]{HKN} $D$ reaches all vertices of $H_1$. 
		
		Let $e=uv$ be an edge of $G$, and let $a_1, a_2, \ldots, a_r$ be the vertices that are added to $e$ in $H_1$. Suppose that $D$ assigns more than one chip to those added vertices, say it assigns one chip to $a_i$ and one to $a_j$ with $i\leq j$. Then we can fire sets $\{a_h \mid i\leq h\leq j\}$, $\{a_h \mid i-1 \leq h\leq j+1\}$, $\ldots$ until at least one of the chips lies on $u$ or $v$. Repeat this procedure untill there is at most one chip assigned to every edge of $G$. The divisor obtained in this way is equivalent to $D$, so it reaches all vertices of $H_1$ and has degree at most $k$. Thus we have obtained a divisor with the desired properties. 
	\end{proof}
	
	Now a certificate can contain the graph $G_1$ and the divisor $D$ as in Lemma \ref{lem:chipsOnRefinement}. From now on we assume $D$ to have chips on vertices of $G_1$ only. 
	A divisor $D$ as in Lemma \ref{lem:chipsOnRefinement} reaches all vertices, so for every vertex $w\in V(G_1)$ there is a divisor $D_w \sim D$ with a chip on $w$ and a level set decomposition $A_1, A_2, \ldots, A_r$ of the transformation $D_w-D$. Again we do not know whether $r$ is polynomial in the size of $G$, so we cannot include this level set decomposition in the certificate. However, we can define some of the sets to be `relevant', and include all relevant sets in the certificate. 
	
	\begin{definition}
		Let $G$ be a graph and $G'$ a subdivision of $G$. Let $D$ be a divisor on $G'$ and $A_1, A_2, \ldots, A_r$ a level set decomposition of a transformation $D' - D$. Let $D_0, D_1, \ldots, D_r$ be the associated sequence of divisors. We call $A_i$ \emph{relevant} if any of the following holds:
		\begin{itemize}
			\item $A_i$ moves a chip from a vertex of $G$, i.e.\ there is a vertex $v$ of $G$ such that $D_i(v) - D_{i-1}(v) < 0$, or
			\item $A_i$ moves a chip to a vertex of $G$, i.e.\ there is a vertex $v$ of $G$ such that $D_i(v) - D_{i-1}(v) > 0$, or
			\item there is a vertex of $G$ such that $A_i$ is the first level set that contains this element, i.e.\ $(A_i\backslash A_{i-1}) \cap V(G)$ is not empty. 
		\end{itemize}
		\label{definition:relevant}
	\end{definition}
	
	\begin{lemma}
	Let $G$ be a graph and $G'$ a subdivision of $G$. Let $D$ be a divisor of degree $k$ on $G'$ and $A_1, A_2, \ldots, A_r$ a level set decomposition of a transformation $D' - D$. Let $D_0, D_1, \ldots, D_r$ be the associated sequence of divisors.
	Then there are at most $2kn+n$ relevant level sets.
	\label{lemma:boundrelevant}
	\end{lemma}
	
	\begin{proof}
	Each chip can reach each vertex at most once and can depart at most once from each vertex. So, there are at most $kn$ sets $A_i$ that fulfill the
	first condition of Definition~\ref{definition:relevant} and at most $kn$ sets that fulfill the second condition. Clearly, the number
	of sets $A_i$ that fulfill the third condition is upper bounded by the number of vertices of $G$.
	\end{proof}

    This lemma shows that the number of  relevant set in a level set decomposition is polynomial.
    However, the number of elements of each of these sets can still be exponential, so we cannot include those sets in a polynomial certificate. Instead, for a relevant set $A_i$, we will include $A_i \cap V(G_1)$ in our certificate. Moreover, for each relevant set, we will descibe which chips move from/to a vertex of $G_1$ by firing $A_i$. When chip $j$ is moved from a vertex $v$ along edge $e$, we include a tuple $(v,j,-1,e)$, and when a chip $j$ is moved to a vertex $v$ along edge $e$, we include a tuple $(v,j,+1,e)$. 
	
	Now, a {\em partial certificate} $\mathcal{C}$ consists of 
	\begin{itemize}
		\item a divisor $D$ of degree $k$ on $G_1$, where the chips are labelled $1, 2, \ldots, k$, 
		\item for every vertex $w$, a series of pairs $(A_{w,1}, M_{w,1})$, 
		$(A_{w,2}, M_{w,2})$, $\ldots$, $(A_{w,a_w},$ $M_{w,a_w})$,
	such that 
	\begin{itemize}
		\item $A_{w,1} \subseteq A_{w,2}\subseteq \ldots \subseteq A_{w,a_w} \subseteq V(G_1)$, 
		\item $M_{w,i} = \{(v,j,\sigma,e) \mid v\in V(G_1), 1\leq j\leq k, \sigma \in \{-1, +1\}, e\in E(G_1)\}$.
	\end{itemize}
	\end{itemize}
	This partical certificate should satisfy a lot of conditions, which are implicit in the intuitive explanation of this partial certificate. We list the intuition of these conditions below and give the formal definition between brackets. 
	\begin{description}
		\item[Incidence requirement] The edge along which a chip is fired is incident to the vertex from/to which it is fired, (i.e.\ for every tuple $(v,j,\sigma,e)$ it holds that $e$ is incident to $v$). 
		\item[Departure requirement] If a chip leaves a vertex, then this vertex is fired and its neighbour is not (i.e.\ if $(v,j,-1,uv) \in M_{w,i}$, then $v\in A_{w,i}$ and $u\notin A_{w,i}$).
		\item[Arrival requirement] If a chip arrives at a vertex, then this vertex is not fired and its neighbour is (i.e.\ if $(v,j,+1,uv) \in M_{w,i}$, then $v\notin A_{w,i}$ and $u\in A_{w,i}$).
		\item[Unique departure per edge requirement] For every vertex at most one chip leaves along each edge (i.e.\ for every $M_{w,i}$, $(v, j_1, -1, e), (v, j_2, -1, e) \in M_{w,i}$ it holds that $j_1 = j_2$).
		\item[Unique arrival per edge requirement] For every vertex at most one chip arrives along each edge (i.e.\ for every $M_{w,i}$, $(v, j_1, +1, e), (v, j_2, +1, e) \in M_{w,i}$ it holds that $j_1 = j_2$).
		\item[Unique departure per chip requirement] A chip can leave a vertex along at most one edge (i.e.\ for every $M_{w,i}$, $(v_1, j, -1, e_1), (v_2, j, -1, e_2) \in M_{w,i}$ it holds that $v_1 = v_2$ and $e_1=e_2$).
		\item[Unique arrival per chip requirement] A chip can arrive at a vertex along at most one edge (i.e.\ for every $M_{w,i}$, $(v_1, j, +1, e_1), (v_2, j, +1, e_2) \in M_{w,i}$ it holds that $v_1 = v_2$ and $e_1=e_2$).
		\item[Immediate arrival requirement] If a chip leaves a vertex $v$ and arrives at another vertex $u$ at the same time, then the chip is fired along the edge $uv$ (i.e.\ for every $M_{w,i}$, if $(v_1, j, -1, e_1), (v_2, j, +1, e_2) \in M_{w,i}$, then $e_1 = e_2 = v_1v_2$).
		\item[Departure location requirement] If a chip leaves a vertex, then this chip was on this vertex (i.e.\ if $(v,j,-1,e) \in M_{w,i}$, then let $i'<i$ be the greatest index such that there is a tuple $(u,j,\sigma,e') \in M_{w,i'}$, then there is a tuple $(v, j, +1, e')\in M_{w,i'}$ for some $e'$. If no such index $i'$ exists, then $D$ assigns $j$ to $v$). 
		\item[Arrival location requirement] If a chip arrives at a vertex, then this chip was moving along an edge to this vertex (i.e.\ if $(v,j,+1,e) \in M_{w,i}$, then either $(u,j,-1,e) \in M_{w,i}$ where $u\neq v$, or let $i'< i$ be the greatest index such that there is a tuple $(u,j,\sigma,e') \in M_{w,i'}$, then there is a tuple $(u, j, -1, e)\in M_{w,i'}$ with $u\neq v$ and $(v,j,+1,e) \notin M_{w,i'}$). 
		\item[Outgoing edges requirement] A chip is fired along each outgoing edge (i.e.\ for every $A_{w,i}$, for every edge $uv$ such that $u\in A_{w,i}$, $v\notin A_{w,i}$, if $(u,j,-1,uv) \notin M_{w,i}$ for all $j$, then there is a $1\leq j\leq k$ and an $i'<i$ such that $(u,j,-1,uv) \in M_{w,i'}$ and $(v,j,+1,uv) \notin M_{w,i''}$ for all $i' \leq i'' < i$). 
		\item[Previous departure requirement] If a chip leaves a vertex $v$ along some edge $e$, and $v$ was in the previous firing set as well, then a chip left $v$ along $e$ when the previous set was fired (i.e.\ if $v\in A_{w,i}$, $v\in A_{w,i+1}$ and $(v,j,-1,e) \in M_{w,i+1}$ for some $j$ and $e$, then $(v,j',-1,e) \in M_{w,i}$ for some $j'\neq j$). 
		\item[Next arrival requirement] If a chip arrives at a vertex $v$ along some edge $e$, and $v$ is not in the next firing set as well, then a chip will arrive at $v$ along $e$ when the next set is fired (i.e.\ if $v\notin A_{w,i}$, $v\notin A_{w,i+1}$ and $(v,j,+1,e) \in M_{w,i}$ for some $j$ and $e$, then $(v,j',+1,e) \in M_{w,i+1}$ for some $j'\neq j$). 
		\item[Reach all vertices requirement] For all vertices $w$, at the end of the sequence $A_{w,1}, \ldots, A_{w,a_w}$, there is a chip on $w$ 
		(i.e.\ for every vertex $w$, either there is a $1\leq j\leq k$ and an $i$ such that $(w,j,+1,e) \in M_{w,i}$ for some $e$ and $(w,j,-1,e') \notin M_{w,i'}$ for all $i' \geq i$, or there is a $1\leq j \leq k$ that $D$ assigns to $w$ and $(w,j,-1,e) \notin M_{w,i}$ for all $i$). 
	\end{description}

	Now for a given graph $G$, and such a partial certificate $\mathcal{C}$, we want to decide whether there is a subdivision
of $G_1$ such that for every vertex $w$ there is a divisor $D_w \sim D$ with a chip on $w$ such that the sets $A_{w,1}, \ldots, A_{w,a_w}$ are the relevant sets of the level set decomposition of the transformation $D_w-D$. To decide this, we will construct an integer linear program $\mathcal{I}_\mathcal{C}$, such that this program has a solution if and only if there is such a subdivision of $G_1$.
Since integer linear programming is in NP, we know that if there is a solution to $\mathcal{I_C}$, then there is a polynomial certificate $\mathcal{D}$ for the ILP instance. In order to obtain a certificate for
for the \textsc{Stable Divisorial Gonality} problem, we add the certificate for the ILP instance to
the partial certificate, as defined above. Thus, 
a certificate for  the \textsc{Stable Divisorial Gonality} problem
is then of the form $(\mathcal{C}, \mathcal{D})$. 
	
	For the integer linear program $\mathcal{I_C}$, we introduce some variables. For every vertex $w \in V(G_1)$ and every $1\leq i< a_w$, we define a variable $t_{w,i}$. This variable represents the number of sets that is fired between $A_{w,i}$ and $A_{w,i+1}$, including $A_{w,i}$ and excluding $A_{w,i+1}$. For every edge $e \in E(G_1)$, we define a variable $l_e$, which represents the length of $e$, i.e.\ the number of edges that $e$ is subdivided into. Now we construct $\mathcal{I_C}$: 
	\begin{itemize}
	\item For every edge $e$, include the inequality $l_{e} \geq 1$. (Every edge has length at least one.)
	\item For every vertex $w$ and $1\leq i< a_w$, include the inequality $t_{w,i} \geq 1$. (The set $A_{w,i}$ is fired, so $t_{w,i} \geq 1$.) 
	\item For every edge $e=uv$ such that there is a set $M_{w,i}$ with $(v,j,-1,e)$, $(u,j,+1,e) \in M_{w,i}$ for some $j$, include $l_e = 1$ in $\mathcal{I_C}$. (If a chip arrives immediately after it is fired, then the edge has length one.)
	\item For every vertex $w$ and $1\leq i< a_w$ such that there are $v,j_1, j_2,e$ such that  $(v,j_1,-1,e) \in M_{w,i}$ and $(v,j_2,-1,e) \in M_{w,i+1}$, include $t_{w,i} = 1$ in $\mathcal{I_C}$. (If there is a set $A$ that is fired between $A_{w,i}$ and $A_{w,i+1}$, then $A_{w,i} \subseteq A \subseteq A_{w,i+1}$. It follows that $A$ fires a chip from $v$ along $e$ as well. But then $A$ is a relevant set. We conclude that $t_{w,i} = 1$.)
	\item For every vertex $w$ and $1\leq i\leq a_w$ such that there are $v,j, e$ such that  $(v,j,+1,e) \in M_{w,i}$, include $t_{w,i} = 1$ in $\mathcal{I_C}$. (Notice that the set fired after $A_{w,i}$ either contains $v$ or causes a chip to arrive at $v$, so this set is relevant.)
	\item For every vertex $w$ and every edge $e = uv$, let $i_0$ be the smallest index such that $(v,j,-1,e) \in M_{w,i_0}$ for some $j$, $i_1$ the greatest index such that $(v,j,-1,e) \in M_{w,i_1}$ for some $j$, $i_2$ the smallest index such that $(u,j,+1,e) \in M_{w,i_2}$ for some $j$ and $i_3$ the greatest index such that $(u,j,+1,e) \in M_{w,i_3}$ for some $j$. Include the following inequalities in $\mathcal{I_C}$: \begin{align}
	(i_1-i_0 +1)l_e - (i_1 - i_0) + (i_3 - i_2) &\geq \sum_{i=i_0}^{i_3} t_{w,i} \label{eq:legroter}\\
	(i_3 - i_2+1) l_e + (i_1 - i_0) - (i_3 - i_2) &\leq \sum_{i=i_0}^{i_3} t_{w,i}. \label{eq:lekleiner}
	\end{align} (There are $i_1-i_0+1$ chips that left $v$ along edge $e$, and $i_3-i_2+1$ chips that arrived at $u$ along $e$. There are $\sum_{i=i_0}^{i_3} t_{w,i}$ sets fired since the first chip left until the last chip arrives, and every of these sets causes one chip to move one step. The chips that arrived at $u$ took $l_e$ steps, the chips that did not arrive took at least one and at most $l_e - 1$ steps. This yields the inequalities.)
	\end{itemize}

	Now a certificate for the stable divisorial gonality problem is a pair $(\mathcal{C}, \mathcal{D})$, where the partial certificate
    $\mathcal{C}$  contains a divisor $D$ on $G_1$ with labeled chips and for every vertex $w$ a series of pairs $(A_{w,1}, M_{w,1}), 
	(A_{w,2}, M_{w,2}), \ldots, (A_{w,a_w},\allowbreak M_{w,a_w})$ that satisfies all requirements above, and where $\mathcal{D}$ is a certificate of the integer linear program $\mathcal{I_C}$.

	\section{Correctness}
	
	It remains to prove that there exists a certificate $(\mathcal{C}, \mathcal{D})$ if and only if $\sdgon(G) \leq k$.

	\begin{lemma}
		Let $G$ be a graph with $\sdgon(G) \leq k$. There exists a certificate $(\mathcal{C}, \mathcal{D})$. 	
        \label{lemma:sdgonimpliceertcertificaat}
	\end{lemma}
	\begin{proof}
		By Lemma \ref{lem:chipsOnRefinement} we know that there is a subdivision $H$ of $G_1$ and a divisor $D$ with $k$ chips, all on vertices of $G_1$, that reaches all vertices. Choose a labeling of the chips and let $D$ be the divisor in $\mathcal{C}$. 
		
		For every vertex $w\in G_1$, there is a divisor $D_w \sim D$ with a chip on $w$ and a level set decomposition $A_{w,1}, \ldots, A_{w,a_w}$. Let $A_{w,i_1}, \ldots, A_{w, i_{b_w}}$ be the subsequence consisting of all relevant sets. Let $B_{w,1} = A_{w,i_1} \cap G_1$, $\ldots$, $B_{w,b_w} = A_{w,i_{b_w}} \cap G_1$. 
		
		Fire the sets $A_{w,1}, \ldots, A_{w,a_w}$ in order. For every $i_j$, set $M_{w,j} = \emptyset$. When firing the set $A_{w, i_j}$, check for every chip $h$ whether it arrives at a vertex $v$ of $G_1$ or leaves a vertex $v$ of $G_1$. If so, add the tuple $(v,h,\sigma,e)$ to $M_{w,j}$, where $\sigma = +1$ if $h$ arrives at $v$ and $\sigma = -1$ is $h$ leaves $v$ and $e$ is the edge of $G_1$ along which $h$ moves. 
		
		The divisor $D$ together with the sequences $(B_{w,i}, M_{w,i})$, for every vertex $w\in V(G_1)$, is the partial certificate $\mathcal{C}$. Notice that by definition $\mathcal{C}$ satisfies all conditions: Incidence requirement, Departure requirement, Arrival requirement, Unique departure per edge, Unique arrival per edge, Unique departure per chip, Unique arrival per chip, Immediate arrival, Departure location, Arrival location, Outgoing edges requirement, Previous departure, Next arrival and Reach all vertices. 
		
		For every edge $e$ of $G_1$, define $l_e$ as the number of edges that $e$ is subdivided into in $H$. For every vertex $w$ of $G_1$ and $1\leq j\leq b_w-1$, define $t_{w,i}$ as the number of sets between $A_{w,i+1}$ and $A_{w,i}$, including $A_{w,i}$ and excluding $A_{w,i+1}$. Notice that this is a solution to the integer linear program $\mathcal{I_C}$. So this is a certificate for this program, write $\mathcal{D}$ for this certificate. Now $(\mathcal{C}, \mathcal{D})$ is a certificate for $(G,k)$. 
	\end{proof}

We illustrate our proof with an example.
    
	\begin{figure}
		\centering
		\begin{subfigure}{0.4 \textwidth}
			\centering
			\begin{tikzpicture}
			\node[vertex, label=below:$u$] (u) at (0,0) {};
			\node[vertex, label=below:$v$] (v) at (2,0) {};
			\draw[edge] (u) to[relative, in=150, out=30] (v);
			\draw[edge] (v) to[relative, in=150, out=30] (u);
			\end{tikzpicture}
			\caption{\footnotesize A graph $G$. }
		\end{subfigure}\quad
	\begin{subfigure}{0.4\textwidth}
		\centering
		\begin{tikzpicture}
		\node[vertex, label=below:$u$, label=$7$] (u) at (0,0) {};
		\node[vertex, label=below:$v$] (v) at (2,0) {};
		\draw[edge] (u) to[relative, in=150, out=30] node [pos=.67, added, label=$x_2$] {} node [pos=.33, added, label=$x_1$] {} (v);
		\draw[edge] (v) to[relative, in=150, out=30] node [midway, added, label=below:$y_2$] {} node [pos=.25, added, label=below:$y_3$] {} node [pos=.75, added, label=below:$y_1$] {} (u);
		\end{tikzpicture}
		\caption{\footnotesize A subdivision of $G$ and divisor.} 
	\end{subfigure}
		\caption{}\label{fig:divisor}
	\end{figure}

\begin{example} \label{ex:example}
Consider the graph in Figure \ref{fig:divisor}. Consider the subdivision in Figure \ref{fig:divisor} and the divisor $D$ with 7 chips on $u$. This divisor reaches $v$, for example by firing the following sets: \begin{align*}
\{u\}, \{u\}, \{u\}, \{u, y_1\}, \{u, x_1, y_1\}, \{u, x_1, y_1\}, \\
\{u, x_1, y_1, y_2\}, \{u, x_1,y_1, y_2\}, \{u, x_1, x_2, y_1, y_2\},\\
\{u, x_1, x_2, y_1, y_2, y_3\}, \{u, x_1, x_2, y_1, y_2, y_3\}, \{u, x_1, x_2, y_1, y_2, y_3\}.
\end{align*}
We describe the corresponding partial certificate $(\mathcal{C}, \mathcal{D})$. The divisor $D$ will be included in $\mathcal{C}$. Notice that there are 8 relevant sets. We obtain the following series of pairs, after labelling the chips $1, 2, \ldots, 7$: \\
\begin{tabular}{lp{.7\textwidth}}
	$A_{v,1} = \{u\}$ & $M_{v,1} = \{(u,1,-1, e_1), (u,2,-1, e_2)\}$ \\ 
	$A_{v,2} = \{u\}$ & $M_{v,2} = \{(u,3,-1, e_1), (u,4,-1, e_2)\}$ \\ 
	$A_{v,3} = \{u\}$ & $M_{v,3} = \{(u,5,-1, e_1), (u,6,-1, e_2)\}$ \\ 
	$A_{v,4} = \{u\}$ & $M_{v,4} = \{(u,7,-1, e_1)\}$ \\ 
	$A_{v,5} = \{u\}$ & $M_{v,5} = \{(v,1,1, e_1)\}$ \\ 
	$A_{v,6} = \{u\}$ & $M_{v,6} = \{(v,3,1, e_1), (v,2,1, e_2)\}$ \\ 
	$A_{v,7} = \{u\}$ & $M_{v,7} = \{(v,5,1, e_1), (v,4,1, e_2)\}$ \\ 
	$A_{v,8} = \{u\}$ & $M_{v,8} = \{(v,7,1, e_1), (v,6,1, e_2)\}$ 
\end{tabular}\\
This gives the partial certificate $\mathcal{C}$. The partial certificate $\mathcal{D}$ consists of a solution to the integer linear program $\mathcal{I_C}$. Here, the corresponding program is: 
\begin{equation*}
\begin{aligned}
l_{e_1} &\geq 1 \\
l_{e_2} &\geq 1 \\
t_{v,i} &\geq 1 \qquad \text{ for } i \in \{1,2,\ldots, 8\}\\
t_{v,i} &=1 \qquad \text{ for } i \in \{1,2,3\} \\
t_{v,i} &=1 \qquad \text{ for } i \in \{5,6,7, 8\} \\
4l_{e_1} &\geq \sum_{i=0}^{8} t_{v,i} 
\end{aligned}
\qquad\qquad\qquad
\begin{aligned}
4l_{e_1} &\leq \sum_{i=0}^{8} t_{v,i} \\
3l_{e_2}-1 &\geq \sum_{i=0}^{8} t_{v,i} \\
3l_{e_2}+1 &\leq \sum_{i=0}^{8} t_{v,i} 
\end{aligned}
\end{equation*}
We can simplify this to:
\begin{equation*}
\begin{aligned}
l_{e_1} &\geq 1 \\
l_{e_2} &\geq 1 \\
t_{v,i} &=1 \qquad \text{ for } i \in \{1,2,3, 5, 6, 7, 8\} \\
t_{v,4} &\geq 1 
\end{aligned}
\qquad\qquad\qquad
\begin{aligned}
4l_{e_1} &\geq  t_{v,4} + 7\\
4l_{e_1} &\leq  t_{v,4} +7 \\
3l_{e_2} &\geq  t_{v,4} +7 \\
3l_{e_2} &\leq  t_{v,4} +7
\end{aligned}
\end{equation*}
We see that $l_{e_1} = 3$, $l_{e_2} = 4$, $t_{v,4} = 5$ and $t_{v,i} = 1$ for $i\neq 4$ is a solution to this program, let this solution be the certificate $\mathcal{D}$. 
\end{example}

	\begin{lemma} \label{lem:certificaatImpliceertSdgon}
		Let $G$ be a graph and $k$ a natural number. Suppose that there exists a certificate $(\mathcal{C}, \mathcal{D})$, then $\sdgon(G) \leq k$. 
	\end{lemma}

	\begin{proof}
		Since $\mathcal{D}$ is a certificate for the integer linear program $\mathcal{I_C}$, we know that there is a solution $l_e$, $t_{w,i}$ for this program. For every edge $e$ of $G_1$, subdivide this edge $l_e -1$ times. Write $H$ for the resulting graph. Now consider the divisor $D$ on $H$. We will prove that $D$ reaches all vertices. 
		
		Let $w$ be a vertex of $G_1$. Consider the set $A_{w,1}, \ldots, A_{w,a_w}$. For every $1\leq i<a_w$, copy $A_{w,i}$ $t_{w,i}$ times, we obtain a sequence \begin{align*}
		A_{w,1,1}, \ldots, A_{w,1,{t_{w,1}}}, A_{w,2,1}, \ldots, A_{w,a_w-1,{t_{w,a_w-1}}}, A_{w,a_w,1}. 
		\end{align*}
		
		For every edge $e = uv$ and every set $A_{w,i,j}$ such that $u,v\in A_{w,i,j}$, add all vertices that are added to $e$ to $A_{w,i,j}$. 
		
		For every edge $e = uv$, let $i_0$ be the smallest index such that $(v,j,-1,e) \in M_{w,i_0}$ for some $j$, $i_1$ the greatest index such that $(v,j,-1,e) \in M_{w,i_1}$ for some $j$, $i_2$ the smallest index such that $(u,j,+1,e) \in M_{w,i_2}$ for some $j$ and $i_3$ the greatest index such that $(u,j,+1,e) \in M_{w,i_3}$ for some $j$. Notice that, by the Previous departure requirement, for all $i_0\leq i<i_1$ there is a chip $j$ such that $(v,j,-1,e) \in M_{w,i}$. It follows that $t_{w,i} = 1$ for all $i_0\leq i<i_1$, since we have the equalities $t_{w,i} = 1$ in $\mathcal{I_C}$ whenever $(v,j_1,-1,e) \in M_{w,i}$ and $(v,j_2,-1,e) \in M_{w,i+1}$. Analogously, for all $i_2 \leq i\leq i_3$ it holds that there is a $j$ such that $(u,j,+1,e) \in M_{w,i}$ and $t_{w,i} = 1$. 
		
		Write $p = i_1 - i_0 +1$, $q = i_3-i_2+1$ and $r = \sum_{i=i_0}^{i_3} t_{w,i} - ql_e$. Let $x_{1}, \ldots, x_{l_e-1}$ be the vertices that are added to $e$, where $x_{1}$ is the neighbour of $v$ and $x_{l_e-1}$ the neighbour of $u$. Consider the sets \begin{gather*}
		A_{w,i_0,1}, A_{w,i_0+1,1}, \ldots, A_{w,i_1,1}, \ldots, A_{w,i_1,{t_{w,i_1}}}, \\
		A_{w,i_1+1,1}, \ldots, A_{w,i_2-1,{t_{w,i_2-1}}}, A_{w,i_2,1}, A_{w, i_2+1,1}, \ldots, A_{w,i_3,1}. 
		\end{gather*} Notice that this are exactly $\sum_{i=i_0}^{i_3} t_{w,i}$ sets. 
		Because of the Departure requirement, we know that $A_{w,i_0,1}$ does contain $v$, so $v$ is an element of all these sets. Analogously, because of the Arrival requirement, we know that $u$ is not an element of $A_{w,i_3,1}$, thus $u$ is not in all those sets. 
		Now we will add the vertices $x_1, \ldots, x_{l_e-1}$ to these sets as follows:\\ 
		Add none of the $x_i$'s to the first $p$ sets, \\
		add $x_1$ to the next $p$ sets, \\
		add $x_1$ and $x_2$ to the next $p$ sets, \\
		$\ldots$, \\
		add $x_1, x_2, \ldots, x_{\lfloor \frac{r}{p-q} \rfloor - 1}$ to the next $p$ sets, \\
		add $x_1, x_2, \ldots, x_{\lfloor \frac{r}{p-q} \rfloor}$ to the next $q + r - (p-q) \lfloor \frac{r}{p-q} \rfloor$ sets, \\
		add $x_1, x_2, \ldots, x_{\lfloor \frac{r}{p-q} \rfloor+1}$ to the next $q$ sets, \\
		$\ldots$, \\
		add $x_1, x_2, \ldots, x_{l_e-1}$ to the next $q$ sets. 
		
		First we need to check whether this adding of vertices is well-defined, i.e.\ check whether $1\leq \lfloor \frac{r}{p-q} \rfloor \leq l_e-1$ and whether the number of sets we add vertices to equals the number of sets we considered: $\sum_{i=i_0}^{i_3} t_{w,i}$.
		Notice that by definition $r\geq p-q$, since $\mathcal{I_C}$ contains the inequality $(i_3 - i_2 +1) l_e + (i_1 - i_0) - (i_3 - i_2) \leq \sum_{i=i_0}^{i_3} t_{w,i}$, or equivalently $q l_e + p - q \leq \sum_{i=i_0}^{i_3} t_{w,i}$. So $\lfloor \frac{r}{p-q} \rfloor \geq 1$. Because of the inequality $(i_1-i_0+1)l_e - (i_1 - i_0) + (i_3 - i_2) \geq \sum_{i=i_0}^{i_3} t_{w,i}$, or equivalently, $pl_e - p + q \geq \sum_{i=i_0}^{i_3} t_{w,i}$, it holds that $\lfloor \frac{r}{p-q} \rfloor \leq l_e-1$. 
		The total number of sets we have added vertices to is: \begin{align*}
		&\quad p\left(\left\lfloor \frac{r}{p-q} \right\rfloor\right) + q + r - (p-q) \left\lfloor \frac{r}{p-q} \right\rfloor + q\left(l_e-1- \left\lfloor \frac{r}{p-q} \right\rfloor\right) \\
		&= q + r + ql_e-q\\
		&= \sum_{i=i_0}^{i_3} t_{w,i}, 
		\end{align*}
		which is exactly the amount of sets we considered. So adding these vertices to the sets is well-defined.

		Adding these vertices yields an increasing sequence of sets. Notice that firing these sets causes $p$ chips to leave $v$: the first $p$ sets $A_{w,i_0,1}, \ldots, A_{w,i_1,1}$ all move a chip from $v$ to $x_1$. And for all $i_0\leq i\leq i_1$, there is a tuple $(v,j,-1,e)\in M_{w,i}$ for some $j$. So firing $A_{w,i,1}$ moves a chip as described by the tuple $(v,j,-1,e)$. 
		
		Notice that if it holds that $\lfloor \frac{r}{p-q} \rfloor = l_e-1$, then it holds that $q + r - (p-q) \lfloor \frac{r}{p-q} \rfloor = q$, so $x_{l_e-1}$ is added to $q$ sets. It follows that the last $q$ sets $A_{w,i_2,1}, \ldots, A_{w,i_3,1}$ cause a chip to arrive at $u$. And for all $i_2\leq i\leq i_3$, there is a tuple $(u,j,+1,e)\in M_{w,i}$ for some $j$. So firing $A_{w,i,1}$ moves a chip as described by the tuple $(u,j,+1,e)$. 
		
		For all $i_1 < i< i_2$ there is no tuple $(x,j,\sigma,e)$ in $M_{w,i}$, and indeed, no chip moves from $v$ along $e$ or to $u$ along $e$ by firing $A_{w,i,k}$.

		Now consider the sets \begin{align*}
		A_{w,1,1}, \ldots, A_{w,1,{t_{w,1}}}, A_{w,2,1}, \ldots, A_{w,a_w-1,{t_{w,a_w-1}}}, A_{w,a_w,1}. 
		\end{align*} after adding the vertices for all edges. By induction we can prove that all chips are moved as described by the tuples in $M_{w,i}$.

		First consider $A_{w,1,1}$. For every tuple $(v,j,-1,uv) \in M_{w,1}$, we know that $v\in A_{w,1,1}$ and $u\notin A_{w,1,1}$. The Outgoing edges requirement implies that along each outgoing edge a chip moves. All chips start on vertices of $G_1$, so for every edges $uv\in E(G_1)$ with $v\in A_{w,1,1}$, $u\notin A_{w,1,1}$ there is a tuples $(v,j,-1,uv) \in M_{w,1}$ for some $j$. From the Unique departure per edge requirement it follows that for every edge $uv$ with $v\in A_{w,1,1}$, $u\notin A_{w,1,1}$ there is exactly one tuple $(v,j,-1,uv) \in M_{w,1}$. The Departure location requirement gives that chip $j$ lies on vertex $v$, and the Unique departure per chip requirement tells us that for every chip $j$ there is at most one edge $e$ such that $(u,j,-1,e) \in M_{w,1}$. So by firing $A_{w,1,1}$ chips can leave vertices of $G_1$ as precribed by the tuples in $M_{w,1}$. If for some edge $uv$ there is a tuple $(v,j,+1,uv) \in M_{w,1}$, then the Arrival location requirement gives that $(u,j,-1,uv)\in M_{w,1}$. It follows that the equation $l_{uv} = 1$ is contained in $\mathcal{I_C}$, so the edge $uv$ is not subdivides in $H$, and indeed chip $j$ will arrive at $v$ by firing $A_{w,1,1}$. If for some edge $uv$ with $u\in A_{w,1,1}$, $v\notin A_{w,1,1}$ there is no tuple $(v,j,+1,uv) \in M_{w,1}$, then it follows from equation \ref{eq:legroter} that $l_{uv} \geq 2$, so indeed, the chip that left $u$ via $uv$ by firing $A_{w,1,1}$ will not arrive at $v$ immediately. We conclude that chips arrive at vertices of $G_1$ as precribed by the tuples in $M_{w,1}$.  
		
		Now consider the set $A_{w,i,j}$. Suppose that for all set $A_{w,i',j'}$ with $i'<i$, or $i' = i$ and $j'<j$, all chips have moved as described by the tuples in $M_{w,i'}$. Now we consider two cases:
		
		First, suppose that $A_{w,i,j}$ is not a relevant set, i.e.\ $j\neq 1$. 
		We know that $A_{w,i,1}$ is the last relevant set before $A_{w,i,j}$. Because of the Outgoing edges requirement, we know that when we fired $A_{w,i,j}$ a chip moved along each edge. If a chip arrived at a vertex of $G_1$, then we know that $t_{w,i} = 1$. This yields a contradiction with the fact that $j\neq 1$. So we conclude that along each outgoing edge of $A_{w,i,j}$ a chip moves. By the construction of $A_{w,i,j}$, we know that no chip will be moved from or to a vertex of $G_1$ by firing $A_{w,i,j}$. 
		
		Now suppose that $A_{w,i,j}$ is a relevant set, i.e.\ $j = 1$. The Outgoing edges requirement implies that a chip is moved along each outgoing edge. Suppose that $(u,j,-1,e)\in M_{w,i}$. Because of the Unique departure per edge requirement, we know that there is no other tuple $(u,j',-1,e)$ in $M_{w,i}$ and because of the Unique departure per chip requirement there is no other tuple $(u,j,-1,e')$. Now the induction hypothesis and the Departure location requirement imply that chip $j$ is on vertex $u$ before firing $A_{w,i,j}$. We conclude that chip $j$ can be moved as described by the tuple $(u,j,-1,e)$. Analogouly, chips can be moved as described by the tuples $(u,j,+1,e)$ in $M_{w,i}$. 
		
		When we fired all sets $A_{w,i,j}$, all chips have moved as described by the tuples in the sets $M_{w,i}$. The Reach all vertices requirement tells us that there is a chip on $w$, so $D$ reaches $w$. 
		
		So $D$ reaches all vertices of $G_1$. Let $x$ be a vertex which is added to an edge $uv$ of $G_1$. Either a chip is fired along $e$ to reach $u$ or $v$, or there is a divisor $D_{uv} \sim D$ with a chip on $u$ and a chip on $v$. In both cases we see that $D$ reaches $x$. So $D$ reaches all vertices. We conclude that $\sdgon(G) \leq k$. 
		\end{proof}

	\begin{example}
		Again consider the graph in Figure \ref{fig:divisor} and the certificate in Example \ref{ex:example}. We can use the construction in the proof of Lemma \ref{lem:certificaatImpliceertSdgon}. Since $l_{e_1} = 3$, we subdivide $e_1$ with two vertices $x_1$ and $x_2$ and since $l_{e_2} = 4$, we subdivide $e_2$ with three vertices $y_1$, $y_2$ and $y_3$. 
		
		For $e_1$ we get $i_0 = 1$, $i_1 = 4$, $i_2 = 5$ and $i_3 = 8$. It follows that $p = 4$, $q=4$ and $r= 12 - 4\cdot 3 = 0$. Now adding the vertices $x_1$ and $x_2$ to the sets yields: \begin{align*}
		A_{v,1,1} = \{u\}, A_{v,2,1} = \{u\}, A_{v,3,1} = \{u\}, A_{v,4,1} = \{u\}, A_{v,4,2} = \{u, x_1\}, \\ A_{v,4,3} = \{u, x_1\},
		A_{v,4,4} = \{u, x_1\}, A_{v,4,5} = \{u, x_1\}, A_{v,5,1} = \{u, x_1, x_2\},\\
		A_{v,6,1} = \{u, x_1, x_2\}, A_{v,7,1} = \{u, x_1, x_2\}, A_{v,8,1} = \{u, x_1, x_2\}.
		\end{align*}
		Analogously for $e_2$, we add the vertices $y_1$, $y_2$ and $y_3$:  \begin{align*}
		\{u\},  \{u\}, \{u\},  \{u, y_1\},
		\{u, x_1, y_1\},  \{u, x_1, y_1\}, \\
		\{u, x_1, y_1,y_2\}, 
		 \{u, x_1,y_1, y_2\}, 
		\{u, x_1, x_2,y_1, y_2\}, \\
		 \{u, x_1, x_2,y_1, y_2, y_3\}, 
	 \{u, x_1, x_2,y_1, y_2, y_3\},
		 \{u, x_1, x_2,y_1, y_2, y_3\}.
		\end{align*}
		We see that we obtained the same subdivision and firing sets as we started with in Example \ref{ex:example}. 
	\end{example}

    	As ILP's have certificates with polynomially many bits (see e.g., \cite{Papadimitriou81}), and the 
partial certificate is of polynomial size (see also Lemma~\ref{lemma:boundrelevant}), we have that, using Lemmas~\ref{lemma:sdgonimpliceertcertificaat} and \ref{lem:certificaatImpliceertSdgon},
the problem whether a given graph has divisorial gonality at most a given integer $k$ has a polynomial
certificate, which gives our main result.

\begin{theorem}
\textsc{Stable Divisorial Gonality} belongs to the class NP.
\end{theorem}

Combined with the NP-hardness by Gijswijt~\cite{Gijswijt}, this yields the following theorem.

\begin{theorem}
\textsc{Stable Divisorial Gonality} is NP-complete.
\end{theorem}

	\section{A bound on subdivisions}
	
	In this section, we give as corollary of our main result a bound on the number of subdivisions needed.
	We use the following result by Papadimitriou~\cite{Papadimitriou81}.
	
	\begin{theorem}[Papadimitriou~\cite{Papadimitriou81}]
	Let $A$ be an $m \times n$ matrix, and $b$ be a vector of length $m$, such that each value in $A$ and $b$ is an integer in the
	interval $[-a,+a]$. If $Ax=b$ has a solution with all values positive integers, then $Ax=b$ has a solution with all values
	positive integers that are at most $n(ma)^{2m+1}$.
	\label{theorem:papadimitriou}
	\end{theorem}
	
	\begin{corollary}
	Let $G$ be a graph with stable divisorial gonality $k$.
	There is a graph $H$, that is a subdivision of $G$, with the divisorial gonality of $H$ equal to the
	stable divisorial gonality of $G$, and each edge in $H$ is obtained by subdividing an edge from $G$ at most
	$m^{O(km^2)}$ times.
	\end{corollary}
	
	\begin{proof}
	By the results of the previous section, we know that there is a certificate whose corresponding ILP has a solution.	
	The values
	$l_e$ in this solution give the number of subdivisions of edges in $G_1$. 
    If we have an upper bound on the number of subdivisions
	of edges in $G'$, say $\alpha$, then $2\alpha+1$ is an upper bound on the number of edges in $G$. 
    Applying Theorem~\ref{theorem:papadimitriou} to the ILP gives such a bound, as described below.

	The ILP has at most $n' \cdot (2kn'+n')$ variables of the
	form $t_{w,i}$, by Lemma~\ref{lemma:boundrelevant}, and $m'$ variables of the form $l_e$, with $n'$ the number of vertices
	in $G_1$ and $m'$ the number of edges in $G_1$. We have $n'=n+m$, and $m'=2m$, with $n'$ the number of vertices of $G$ and $m$
	the number of edges of $G$.
	
	The number of equations and inequalities in the ILP is linear in the number of variables.
An inequality can be replaced by an equation by adding one variable. This gives a total of
$O(kn'^2+m')$ variables and $O(kn'^2+m')$ equations. Note that $O(kn'^2+m') = O(km^2)$; as $G$ is
connected, $n\leq m-1$.
Also, note all values in matrix $A$ and vector $b$
are $-1$, $0$, or $1$, i.e., we can set $a=1$ in the application of Theorem~\ref{theorem:papadimitriou}.
So, by Theorem~\ref{theorem:papadimitriou}, we obtain that if there is a solution to the ILP, then there
is one where all variables are set to values at most 
\[ O(kn'^2+m') \cdot O(kn'^2+m')^{O(kn'^2+m')} = O(km^2) \cdot O(km^2)^{O(km^2)} = m^{O(km^2)} \]
	
	When taking $k$ the stable divisorial gonality of $G$, we know there is at least one certificate
    with a solution, so we can bound the number of subdivisions in $G_1$ by $m^{O(km^2)}$, which
    gives our result.
	\end{proof}

	\section{Conclusion}
	In this paper, we showed that the problem to decide whether the stable divisorial gonality of a given graph is at most a given number
$K$ belongs to the class NP. Together with the NP-hardness result of Gijswijt~\cite{Gijswijt}, this shows that the problem is NP-complete.
We think our proof technique is interesting: we give a certificate that describes some of the essential aspects of the firing sequences; whether there is a subdivision of the graph for which this certificate
describes the firing sequences and thus gives the subdivision that reaches the optimal divisorial gonality
can be expressed in an integer linear program. Membership in NP then follows by adding the certificate
of the ILP to the certificate for the essential aspects.

As a byproduct of our work, we obtained an upper bound on the number of subdivisions needed to reach
a subdivision of $G$ whose divisorial gonality gives the stable divisorial gonality of $G$. Our 
upper bound still is very high, namely exponential in a polynomial of the size of the graph. An
interesting open problem is whether this bound on the number of needed subdivisions can be replaced by a polynomial in the size of the graph. Such a result would give an alternative (and probably easier) proof
of membership in NP: first guess a subdivision, and then guess the firing sequences.

There are several open problems related to the complexity of computing the (stable) divisorial gonality of graphs. Are these problems
{\em fixed parameter tractable}, i.e., can they be solved in $O(f(K)n^c)$ time for some function $k$ and constant $c$? Or can they be proven to be
$W[1]$-hard, or even, is there a constant $c$, such that deciding if (stable) divisorial gonality of a given graph $G$ is at most $c$ is
already NP-complete? Also, how well can we approximate the divisorial gonality or stable divisorial gonality of a graph?

	\subsection*{Acknowledgements} We thank Gunther Cornelissen and Nils Donselaar for helpful discussions.


\end{document}